\documentclass[copyright]{eptcs}
\providecommand{\event}{Express 2011} 
\usepackage{breakurl}             
\usepackage{url}
\usepackage{amssymb,amsmath}
\usepackage{bm}
\usepackage{amsthm}
\usepackage[dvips]{graphics,graphicx}
\usepackage[normalem]{ulem}
\usepackage{psfrag}
\newtheorem{theorem}{Theorem}[section]
\newtheorem{lemma}[theorem]{Lemma}
\newtheorem{proposition}[theorem]{Proposition}

\newtheorem{conjecture}[theorem]{Conjecture}

\theoremstyle{definition}
\newtheorem{definition}[theorem]{Definition}
\newtheorem{example}[theorem]{Example}
\theoremstyle{remark}
\newtheorem{remark}[theorem]{Remark}

\newtheorem{notation}[theorem]{Notation}
\newtheoremstyle{TheoremRepeat}
    {\topsep}{\topsep}              
    {\itshape}                      
    {}                              
    {\bfseries}                     
    {.}                             
    { }                             
    {\thmname{#1}\thmnote{ \bfseries #3}}
\theoremstyle{TheoremRepeat}
\newtheorem{theoremrep}{Theorem}
\newtheorem{lemmarep}{Lemma}
\newtheorem{propositionrep}{Proposition}
\newcommand{\Comment}[1]{}
\newcommand{\Defeq}{\stackrel{\mathrm{df}}{=}}
\newcommand{\bnfeq}{::=}
\newcommand{\card}[1]{|#1|}
\newcommand{\bsep}{\mid}
\newcommand{\union}{\cup}
\newcommand{\Union}{\bigcup}
\newcommand{\inter}{\cap}
\newcommand{\cross}{\times}
\newcommand{\pow}{\mc P}
\newcommand{\res}{\restriction}
\newcommand{\Nat}{\mathbb{N}}
\newcommand{\dom}[1]{\mathsf{dom}(#1)}
\newcommand{\rge}[1]{\mathsf{rge}(#1)}
\newcommand{\Par}{\mathrel{\mbox{$\! \mid\! $}}}
\newcommand{\m}[1]{\mbox{#1}}
\newcommand{\mc}[1]{\mathcal{#1}}
\newcommand{\lab}{\ell}
\newcommand{\Act}{\mathsf{Act}}
\newcommand{\co}{\mathrel{co}}
\newcommand{\tran}[1]{\stackrel{#1}{\rightarrow}}
\newcommand{\tranc}[1]{\stackrel{#1}{\rightarrow}_{\mc C}}
\newcommand{\trand}[1]{\stackrel{#1}{\rightarrow}_{\mc D}}
\newcommand{\rtran}[1]{\stackrel{#1}{\rightsquigarrow}}
\newcommand{\rtranc}[1]{\stackrel{#1}{\rightsquigarrow}_{\mc C}}
\newcommand{\rtrand}[1]{\stackrel{#1}{\rightsquigarrow}_{\mc D}}
\newcommand{\fd}[1]{\langle{#1}\rangle\! \rangle}
\newcommand{\dec}[1]{({#1})}
\newcommand{\di}[1]{\langle{#1}\rangle}
\newcommand{\rd}[1]{\langle\! \langle{#1}\rangle}
\newcommand{\fb}[1]{\mathrel{\, [{#1}]] }}
\newcommand{\rb}[1]{\mathrel{\, [[{#1}]}}
\newcommand{\Fi}[1]{\mathsf{fi}(#1)}

\newcommand{\ttrue}{\mathrm{t\! t}}
\newcommand{\ffalse}{\m{\rm ff}}

\newcommand{\oor}{\vee}
\newcommand{\aand}{\wedge}
\newcommand{\Oor}{\bigvee}
\newcommand{\Aand}{\bigwedge}
\newcommand{\depth}[1]{\mathsf{md}(#1)}

\newcommand{\EIL}[1]{\mathrm{EIL}_{\mathrm{#1}}}
\newcommand{\eql}[1]{\sim_{#1}}

\newcommand{\eqb}[1]{\approx_{\mathsf{#1}}}

\newcommand{\isom}{\cong}

\newcommand{\Id}{\mathsf{Id}}
\mathchardef\mh="2D

\newcommand{\Char}[2]{\chi_{#2}^{\mathrm{#1}}}
\newcommand{\Chard}[3]{\chi_{#2,#3}^{\mathrm{#1}}}

\newcommand{\fbc}[4]{(\boldsymbol{#1},\bar{\boldsymbol{#2}} < \mathsf{#3}\, #4 )}

\title{A Logic with Reverse Modalities for History-preserving Bisimulations}
\author{Iain Phillips
\institute{Department of Computing, Imperial College London, England}
\email{iccp@doc.ic.ac.uk}
\and
Irek Ulidowski
\institute{Department of Computer Science, University of Leicester, England}
\email{iu3@mcs.le.ac.uk}
}
\def\titlerunning{A Logic with Reverse Modalities}
\def\authorrunning{I.C.C. Phillips \& I. Ulidowski}
\begin{document}
\maketitle

\begin{abstract}
We introduce event identifier logic (EIL) which extends Hennessy-Milner logic by the addition of (1) reverse as well as forward modalities, and (2) identifiers to keep track of events.
We show that this logic corresponds to hereditary history-preserving (HH) bisimulation equivalence within a particular true-concurrency model, namely stable configuration structures.
We furthermore show how natural sublogics of EIL correspond to coarser equivalences.
In particular we provide logical characterisations of weak history-preserving (WH) and history-preserving (H) bisimulation.
Logics corresponding to HH and H bisimulation have been given previously, but
not to WH bisimulation (when autoconcurrency is allowed), as far as we are aware.
We also present characteristic formulas which characterise individual structures
with respect to history-preserving equivalences.
\end{abstract}

\section{Introduction}

The paper presents a modal logic that can express simple properties of computation in 
the true concurrency setting of stable configuration structures. We aim, like 
Hennessy-Milner logic (HML)~\cite{HM85} in the interleaving setting, to characterise
the main true concurrency equivalences and to develop characteristic formulas for them. 
We focus in this
paper on history-preserving bisimulation equivalences.

HML has a ``diamond'' modality $\di a \phi$
which says that an event labelled $a$ can be performed,
taking us to a new state which satisfies $\phi$.
The logic also contains negation ($\neg$), conjunction ($\aand$)
and a base formula which always holds ($\ttrue$).
HML is strong enough to distinguish any two processes which are not bisimilar.

We are interested in making true concurrency distinctions between processes.
These processes will be {\em event structures},
where the current state is represented by the set of events which have occurred so far.
Such sets are called {\em configurations}.
Events have labels (ranged over by $a,b,\ldots$), and different events may have the same 
label.  We shall refer to example event structures using a CCS-like notation,
with $a \Par b$ denoting an event labelled with $a$ in parallel with another labelled 
with $b$,
$a . b$ denoting two events labelled $a$ and $b$ where the first causes the second,
and $a + b$ denoting two events labelled $a$ and $b$ which conflict.

In the true concurrency setting bisimulation is referred to as 
\emph{interleaving bisimulation}, or IB for short.
The processes $a \Par b$ and $a.b + b.a$ are interleaving bisimilar,
but from the point of view of true concurrency they should be distinguished,
and HML is not powerful enough to do this.

We therefore look for a more powerful logic, and we base this logic on adding reverse moves.
Instead of the one modality $\di a \phi$ we have two:
{\em forward diamond} $\fd a \phi$
(which is just a new notation for the $\di a \phi$ of HML)
and {\em reverse diamond} $\rd a \phi$.  
The latter is satisfied if we can reverse some event labelled with $a$ and get to a configuration where $\phi$ holds.
Such an event would have to be {\em maximal} to enable us to reverse it,
i.e.\ it could not be causing some other event that has already occurred.

With this new reverse modality we can now distinguish  $a \Par b$ and $a.b + b.a$:
$a \Par b$ satisfies $\fd a \fd b \rd a \ttrue$, while $a.b + b.a$ does not.
The formula expresses the idea that $a$ and $b$ are {\em concurrent}.
Alternatively we see that 
$a.b + b.a$ satisfies $\fd a \fd b \neg \rd a \ttrue$, while $a \Par b$ does not.
This latter formula expresses the idea that $a$ {\em causes} $b$.

The new logic corresponds to \emph{reverse interleaving bisimulation} \cite{PU11}, 
or RI-IB for short.  In the absence of autoconcurrency,
Bednarczyk~\cite{Bed91} showed that this is as strong as 
\emph{hereditary history-preserving bisimulation}~\cite{Bed91}, or HH for short,
which is usually regarded as the strongest desirable true concurrency equivalence.
HH was independently proposed in~\cite{JNW96}, under the name of strong history-preserving bisimulation.

Auto-concurrency is where events can occur concurrently and have the same label.
To allow for this, we need to strengthen the logic.
For instance, we want to distinguish $a \Par a$ from $a.a$,
which is not possible with the logic as it stands:
$\fd a \fd a \rd a \ttrue$ is satisfied by both processes.
We need some way of distinguishing the two events labelled with $a$.
We change our modalities so that when we make a forward move we {\em declare} an 
{\em identifier} (ranged over by $x,y,\ldots$) which stands for that event,
allowing us to refer to it again when reversing it. 
Now we can write $\fd {x:a} \fd{y:a} \rd x \ttrue$,
and this is satisfied by $a \Par a$, but not by $a.a$.
Declaration is an identifier-binding operation, so that $x$ and $y$ are both bound in the formula.
Baldan and Crafa~\cite{BC10} also used such declarations in their forward-only logic.

With this simple change we now have a logic which is as strong as HH,
even with autoconcurrency.

We have to be careful that our logic does not become too strong.
For instance, we want to ensure that processes $a$ and $a + a$ are indistinguishable.
One might think that $a+a$ satisfies $\fd{x:a} \rd x \fd{y:a} \neg \rd x \ttrue$,
while $a$ does not.
To avoid this, we need to ensure that $x$ is forgotten about once it is reversed,
and so cannot be used again.
One could make a syntactic restriction that in a formula $\rd x \phi$ the identifier $x$ is not allowed to occur (free) in $\phi$.
However this is not actually necessary, as our semantics will ensure that all identifiers must be assigned to events in the current configuration.
So in fact $\fd{x:a} \rd x \fd{y:a} \neg \rd x \ttrue$ is not satisfied by $a+a$,
since we are not allowed to reverse $x$ as it would take us to a configuration where
$x$ is mentioned in $\fd{y:a} \neg \rd x \ttrue$ but $x$ is assigned to an event outside the current configuration.
Baldan and Crafa~\cite{BC10} also had to deal with this issue.

Our logic is not quite complete,
since we wish to express certain further properties.
For instance, we would like to express a reverse move labelled with $a$,
i.e.\ $\rd a \phi$.
Instead of adding this directly, we add {\em declarations} $\dec{x:a}\phi$.
We can now express $\rd a \phi$ by the formula
$\dec{x:a}\rd x \phi$ (where $x$ does not occur (free) in $\phi$).

We also wish to express so-called {\em step transitions}, which are transitions consisting of multiple events occurring concurrently.  For instance a forward step $\fd{a,a} \phi$ of two events labelled with $a$ can be achieved by
$\fd {x:a}\fd{y:a}(\phi \aand \rd x \ttrue)$ and a reverse step $\rd{a,a} \phi$ can be achieved by
$\dec{x:a}\dec{y:a}(\rd x \rd y \phi \aand \rd y \ttrue)$
(both formulas with $x$ and $y$ not free in $\phi$).
Thus the reverse steps employ declarations.
As well as expressing reverse steps, declarations allow us to obtain a sublogic 
which corresponds to \emph{weak history-preserving bisimulation} (WH).

\begin{figure}
\psfrag{HH}{HH}
\psfrag{H}{H}
\psfrag{HWH}{HWH}
\psfrag{WH}{WH}
\centering
\includegraphics[width=3cm]{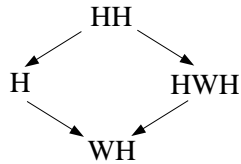}
\caption{The hierarchy of history-preserving equivalences.}\label{fig:Eqs}
\end{figure}

This completes a brief introduction of our logic,
which we call {\em Event Identifier Logic}, or $\EIL{}$ for short. 
Apart from corresponding to HH, EIL has natural sublogics for several 
other true concurrency equivalences.
Figure~\ref{fig:Eqs} shows a hierarchy of equivalences that we are able to
characterise,
where arrows denote proper set inclusion.  Apart from the mentioned HH and WH,
\emph{history-preserving bisimulation} (H) is a widely studied equivalence 
that employs history isomorphism.
\emph{Hereditary weak-history
preserving bisimulation} (HWH) is WH with the hereditary property \cite{Bed91}
that deals with reversing
of events. The definitions 
of these equivalences can be found
in \cite{vGG01,PU11}, and are outlined in Section~\ref{subsec:equivalences}.

It is natural to ask if, at least for a finite structure, there is a single logical formula which captures all of its behaviour, up to a certain equivalence.
Such formulas are called {\em characteristic} formulas.
They have been investigated previously for HML and other logics~\cite{GS86,SI94,AIS09}.
We look at characteristic formulas with respect to three of the equivalences we consider, namely HH, H and WH.

The main contribution of the paper is a logic EIL. It could be argued that EIL is a natural
and canonical logic for the true concurrency equivalences considered here
in the following sense.  Firstly, its forward and reverse modalities capture faithfully 
the information of the forward and reverse transitions 
in the definitions of the equivalences, Secondly, event identifier environments and
event declarations give rise naturally to order isomorphisms for HH, H, HWH and WH.
Finally, EIL extends HML and keeps with its spirit of having simple modalities
defined seamlessly over a general computation model. 

Other contributions include the first to 
our knowledge logics for WH and HWH.
Finally, we present the first to our knowledge characteristic formulas for HH, H and WH.

The paper is organised as follows.
We look at related work in Section~\ref{sec:related}.
Then we recall the definitions of configuration structures and the bisimulation-based equivalences that we shall need in Section~\ref{sec:config eqnces}.
We then introduce $\EIL{}$ in Section~\ref{sec:EIL}, giving examples of its usage.
Next we look at how to characterise various equivalences using $\EIL{}$ and its sublogics (Section~\ref{sec:EIL eqnces}).
In Section~\ref{sec:char} we investigate characteristic formulas.
We finish with conclusions and future work.

\section{Related work}\label{sec:related}
Previous work on logics for true concurrency can be categorised loosely according to
the type of semantic structure (model) that the satisfaction relation of the logic
is defined for.  There are logics over 
configurations (sets of consistent events) \cite{GKP92,BC10} and logics over 
paths (or computations) \cite{Che92,NC94,NC94a,NC95,PLS94}, although logics in 
\cite{NC94,NC94a,NC95} can be seen also as logics over configurations.
Other structures such as trees, graphs and Kripke frames are used as models 
in, for example, \cite{DNF90,MT92,Gut09,GB09}. 

The logic in this paper uses simple forward and reverse
event identifier modalities that are sufficient to characterise HH.
In contrast, Baldan and Crafa \cite{BC10}
achieved an alternative characterisation of HH with a different modal logic that uses solely
forward-only event identifier modalities $\di{x}$ and $\fbc{x}{y}{a}{z}$. 
The formula $\fbc{x}{y}{a}{z} \phi$ holds in a configuration if in its future
there is an $\mathsf{a}$-labelled event $e$ that can be bound to $z$, and $\phi$ holds.
Additionally,
$e$ must be (1) caused at least by the events already bound to the events in $\boldsymbol{x}$
and (2) concurrent with at least the events already bound to the events in $\boldsymbol{y}$.
Several interesting sublogics were also identified in \cite{BC10} that characterise
H, pomset bisimulation~\cite{BC87,vGG01} and step bisimulation~\cite{Pom86,vGG01}
respectively.

Goltz, Kuiper and Penczek \cite{GKP92} researched configurations of prime
event structures \emph{without autoconcurrency}. In such a setting HH
coincides with reverse interleaving bisimulation RI-IB
(shown in~\cite{Bed91}).
Moreover, H coincides with WH.
\emph{Partial Order Logic} (POL) is proposed
in \cite{GKP92}. POL contains past modalities and the authors stated that it
characterises RI-IB (and thus HH).  Also, it is conjectured that if one restricts POL 
in such a way that 
no forward modalities can be nested in a past modality, then such a logic characterises H
(and thus WH).

Cherief \cite{Che92} defined a pomset bisimulation relation over paths and
shows that it coincides with H (defined over configurations). The author
then predicted that an extension of HML with forward and reverse pomset modalities
characterises H. This idea was then developed further by 
Pinchinat, Laroussinie and Schnoebelen in \cite{PLS94}. 

Nielsen and Clausen defined a $\delta$-bisimulation relation ($\delta$b) over paths 
\cite{NC94,NC95}.  Unlike in \cite{Che92,PLS94}, one is allowed to reverse  
independent maximal events in any order. This seemingly small change has a profound 
effect on the strength of the equivalence: $\delta$b coincides with HH. 
It was shown that an extension of HML with a reverse modality characterises HH 
when there is no autoconcurrency \cite{NC94,NC95}.
Additionally, it was stated (without a proof) \cite{NC94a}
that an extension of HML with a reverse $\emph{event index}$ modality characterises HH
even in the presence of autoconcurrency. The notion of paths used in \cite{NC94,NC94a,NC95}
induces a notion of configuration. Hence, their  logics could be understood as logics
over configurations and reverse index modality could be seen as a form of our reverse 
event identifier modality. We would argue, however, that many properties of configurations
related to causality and concurrency between events 
are expressed more naturally with reverse identifier modalities.

\Comment{
Gutierrez introduced a modal logic for transition systems with independence 
\cite{Gut09,GB09}.
The logic has two diamond modalities: $\di{a}_c$ (respectively $\di{a}_{nc}$) asserts that,
having performed a transition $t$ to reach the current state, there is a
causally dependent (respectively concurrent) wrt $t$ transition $t'$ with label $a$
that can be performed. The logic is as strong as H.
}

Past or reverse modalities, which are central to our logic, were used before 
in a number of modal logics and temporal logics 
\cite{HS85,DNV90,DNMV90,DNF90,GKP92,LPS95,LS95,Pen95} but only 
\cite{DNF90,GKP92} proposed logical characterisations of true concurrency equivalences.  
Among the rest, HML with backward modalities in \cite{DNV90,DNMV90} defined over paths
is shown to characterise branching bisimulation. Finally, Gutierrez introduced a modal logic 
for transition systems with independence \cite{Gut09,GB09} that has
two diamond modalities: one for causally dependent transitions and the other for
concurrent transitions with respect to a given transition.

\section{Configuration structures and equivalences}\label{sec:config eqnces}

In this section we define our computational model (stable configuration structures) and the various bisimulation equivalences for which we shall present logical characterisations.

\subsection{Configuration structures}\label{subsec:config}
We work with stable configuration structures~\cite{vGP95,vGP09,vGG01}, which are equivalent to stable event structures~\cite{Win87}.


\begin{definition}\label{def:cs}
A {\em configuration structure} (over an alphabet $\Act$) is a pair $\mc C = (C,\lab)$ where $C$ is a family of finite sets (configurations) and $\lab:\Union_{X \in C} X \to \Act$ is a labelling function.
\end{definition}
We use $C_{\mc C}, \lab_{\mc C}$ to refer to the two components of a configuration structure $\mc C$.
Also we let $E_{\mc C} = \Union_{X \in C} X$, the {\em events} of $\mc C$.
We let $e,\ldots$ range over events, and $E,F,\ldots$ over sets of events.
We let $a,b,c,\ldots$ range over labels in $\Act$.
\begin{definition}[\cite{vGG01}]
A configuration structure $\mc C = (C,\lab)$ is {\em stable} if it is
\begin{itemize}
\item
rooted: $\emptyset \in C$;\quad 
connected: $\emptyset \neq X \in C$ implies $\exists e \in X: X \setminus\{e\} \in C$;
\item
closed under bounded unions: if $X,Y,Z \in C$ then $X \union Y \subseteq Z$ implies $X \union Y \in C$;
\item
closed under bounded intersections: if $X,Y,Z \in C$ then $X \union Y \subseteq Z$ implies $X \inter Y \in C$.
\end{itemize}
\end{definition}

Any stable configuration structure is the set of configurations of a 
stable event structure~\cite[Thm~5.3]{vGG01}.

\begin{definition}
Let $\mc C = (C,\lab)$ be a stable configuration structure, and let $X \in C$.
\begin{itemize}
\item
Causality: $d \leq_X e$ iff for all $Y \in C$ with $Y \subseteq X$ we have $e \in Y$ implies $d \in Y$.
Furthermore $d <_X e$ iff $d \leq_X e$ and $d \neq e$.
\item
Concurrency: $d \co_X e $ iff $d \not <_X e$ and $e \not <_X d$.
\end{itemize}
\end{definition}

It is shown in~\cite{vGG01} that $<_X$ is a partial order and that the sub-configurations of $X$ are precisely those subsets $Y$ which are left-closed w.r.t.\ $<_X$, i.e.\ if $d <_X e \in Y$ then $d \in Y$.
Furthermore, if $X,Y \in C$ with $Y \subseteq X$, then ${<_Y} = {<_X \res Y}$.

Recall that a prime event structure is a set of events with 
a labelling function, together with a causality relation and a conflict 
relation (between events that cannot be members of the same 
configuration)~\cite{Win87}.
The set of configurations of a prime event structure forms a stable 
configuration structure; prime event structures are a proper subclass 
of stable event structures.  All of our examples are given as prime event structures 
or the corresponding CCS expressions.
When drawing diagrams of prime event structures we shall, as usual, depict 
the causal relation with arrows, and the conflict relation with dotted lines.  
We shall also suppress
the actual events and write their labels instead.  Thus if we have two events $e_1$ and 
$e_2$, both labelled with $a$, in diagrams we shall denote them as $a_1$ and $a_2$, 
respectively, when we wish to distinguish between them. This is justified, since all 
the notions of equivalence we shall discuss depend on the labels of the events, 
rather than the events themselves.

\begin{example}\label{simple_pes}
Consider a prime event structure with events $e_1, e_2, e_3$ all labelled with
$a$, where $e_1$ causes $e_2$ and $e_1,e_2$ are concurrent with $e_3$. The corresponding CCS
expression is $(a.a) \Par a$. The set of configurations
consists of $\emptyset$, $\{e_1\}, \{e_3\}, \{e_1, e_2\}, \{e_1,e_3\}$ and $\{e_1, e_2,e_3\}$. 

\end{example}

\begin{definition}\label{def:tran}
Let $\mc C = (C,\lab)$ be a stable configuration structure and let $a \in \Act$.
We let $X \tranc e X'$ iff $X,X' \in C$, $X \subseteq X'$ and $X' \setminus X = \{e\}$.
Furthermore we let $X \tranc a X'$ iff $X \tranc e X'$ for some $e$ with $\lab(e) = a$.
We also define reverse transitions:
$X \rtranc e X'$ iff $X' \tranc e X$, and $X \rtranc a X'$ iff $X' \tranc a X$.
The overloading of notation whereby transitions can be labelled with events or with event labels should not cause confusion.
\end{definition}
For a set of events $E$, let $\lab(E)$ be the multiset of labels of events in $E$.
We define a {\em step} transition relation where concurrent events are executed in a single step:
\begin{definition}
Let $\mc C = (C,\lab)$ be a stable configuration structure and let $A \in \Nat^\Act$ ($A$ is a multiset over $\Act$).
We let $X \tranc A X'$ iff $X,X' \in C$, $X \subseteq X'$, and $X' \setminus X = E$ with $d \co_{X'} e$ for all $d,e \in E$ and $\lab(E) = A$.
\end{definition}
We shall assume in what follows that stable configuration structures are {\em image finite} with respect to forward transitions,
i.e.\ for any configuration $X$ and any label $a$, the set $\{X':X \tranc a X'\}$ is finite.

\subsection{Equivalences}\label{subsec:equivalences}

We define history-preserving bisimulations and illustrate the differences between them 
with examples.

\begin{definition}\label{def:isom pomset}
Let $\mc X = (X,<_X,\lab_X)$ and $\mc Y = (Y,<_Y,\lab_Y)$ be partial orders which are labelled over $\Act$.
We say that $\mc X$ and $\mc Y$ are {\em isomorphic} ($X \isom Y$) iff there is a bijection from $X$ to $Y$ respecting the ordering and the labelling.  The isomorphism class $[\mc X]_{\isom}$ of a partial order labelled over $\Act$ is called a {\em pomset} over $\Act$.
\end{definition}

\begin{definition}[\cite{DDNM87,vGG01}]\label{def:wh}
Let $\mc C, \mc D$ be stable configuration structures.
A relation $\mc R \subseteq C_{\mc C} \cross C_{\mc D}$ is a {\em weak history-preserving (WH) bisimulation} between $\mc C$ and $\mc D$ if $\mc R(\emptyset,\emptyset)$ and if $\mc R(X,Y)$ and $a \in \Act$ then:
\begin{itemize}
\item
$(X, <_X,\lab_{\mc C}\res X) \isom (Y, <_Y,\lab_{\mc D}\res Y)$;
\item
if $X \tranc a X'$ then $\exists Y'.\ Y \trand a Y'$ and $\mc R(X',Y')$;
\item
if $Y \trand a Y'$ then $\exists X'.\ X \tranc a X'$ and $\mc R(X',Y')$.
\end{itemize}
We say that $\mc C$ and $\mc D$ are WH equivalent ($\mc C \eqb{wh} \mc D$) iff there is a WH bisimulation between $\mc C$ and $\mc D$.
\end{definition}

\begin{definition}[\cite{RT88,vGG01}]\label{def:h}
Let $\mc C, \mc D$ be stable configuration structures.
A relation $\mc R \subseteq C_{\mc C} \cross C_{\mc D} \cross \pow(E_{\mc C} \cross E_{\mc D})$ is a {\em history-preserving (H) bisimulation} between $\mc C$ and $\mc D$ iff $\mc R(\emptyset,\emptyset,\emptyset)$ and if $\mc R(X,Y,f)$ and $a \in \Act$
\begin{itemize}
\item
$f$ is an isomorphism between $(X, <_X,\lab_{\mc C}\res X)$ and $(Y, <_Y,\lab_{\mc D}\res Y)$;
\item
if $X \tranc a X'$ then $\exists Y',f'.\ Y \trand a Y'$, $\mc R(X',Y',f')$ and 
$f'\res X = f$;
\item
if $Y \trand a Y'$ then $\exists X',f'.\ X \tranc a X'$, $\mc R(X',Y',f')$ and 
$f'\res X = f$.
\end{itemize}
We say that $\mc C$ and $\mc D$ are H equivalent ($\mc C \eqb{h} \mc D$) iff there is 
an H bisimulation between $\mc C$ and $\mc D$.
\end{definition}

Both H and WH have associated hereditary versions:
\begin{definition}[\cite{Bed91,JNW96,vGG01}]\label{def:hh}
Let $\mc C, \mc D$ be stable configuration structures and let $a \in \Act$.
Then $\mc R \subseteq C_{\mc C} \cross C_{\mc D} \cross \pow(E_{\mc C} \cross E_{\mc D})$
is a hereditary H (HH) bisimulation iff $\mc R$ is an H bisimulation and if $\mc R(X,Y,f)$ then
for any $a \in \Act$,
\begin{itemize}
\item
if $X \rtranc a X'$ then $\exists Y',f'.\ Y \rtrand a Y'$, $\mc R(X',Y',f')$  and
$f\res X' = f'$;
\item
if $Y \rtrand a Y'$ then $\exists X',f'.\ X \rtranc a X'$, $\mc R(X',Y',f')$
and $f\res X' = f'$.
\end{itemize}
We say that $\mc C$ and $\mc D$ are HH equivalent ($\mc C \eqb{hh} \mc D$) iff 
there is an HH bisimulation between $\mc C$ and $\mc D$.
\end{definition}

\begin{definition}\label{def:hwh}
Let $\mc C, \mc D$ be stable configuration structures and let $a \in \Act$.
Then $\mc R \subseteq C_{\mc C} \cross C_{\mc D} \cross \pow(E_{\mc C} \cross E_{\mc D})$
is a hereditary WH (HWH) bisimulation if $\mc R(\emptyset,\emptyset,\emptyset)$ and if $\mc R(X,Y,f)$ and $a \in \Act$ then:
\begin{itemize}
\item
$f$ is an isomorphism between $(X, <_X,\lab_{\mc C}\res X)$ and $(Y, <_Y,\lab_{\mc D}\res Y)$;
\item
if $X \tranc a X'$ then $\exists Y',f'.\ Y \trand a Y'$ and $\mc R(X',Y',f')$;
\item
if $Y \trand a Y'$ then $\exists X',f'.\ X \tranc a X'$ and $\mc R(X',Y',f')$;
\item
if $X \rtranc a X'$ then $\exists Y',f'.\ Y \rtrand a Y'$, $\mc R(X',Y',f')$  and
$f\res X' = f'$;
\item
if $Y \rtrand a Y'$ then $\exists X',f'.\ X \rtranc a X'$, $\mc R(X',Y',f')$
and $f\res X' = f'$.
\end{itemize}
Also $\mc C$ and $\mc D$ are HWH equivalent ($\mc C \eqb{hwh} \mc D$) iff 
there is an HWH bisimulation between $\mc C$ and $\mc D$.
\end{definition}
The inclusions in Figure~\ref{fig:Eqs} are immediate from the definitions.
They are strict inclusions:

\begin{example}[\cite{PU11}]\label{ex:hwh-not-h}
Consider event structures $\mc E$, $\mc F$ in Figure~\ref{fig:hwh-not-h}, where each
event structure has four $a$-labelled and four $b$-labelled events.
$\mc E = \mc F$ holds for $\eqb{hwh}\,$, and hence for $\eqb{wh}\,$, but not for $\eqb{h}\,$,
and hence not for $\eqb{hh}\,$. We now show this.  $\mc E$, $\mc F$ have the same
configurations except that $\{a_2,a_3,b_3\}$ is missing in $\mc F$. We define a bisimulation
by relating all isomorphic states, and check that it is an HWH.
\begin{figure}
\centering
\psfrag{E}{$\mc E$}
\psfrag{F}{$\mc F$}
\psfrag{a1}{$a_1$}
\psfrag{a2}{$a_2$}
\psfrag{a3}{$a_3$}
\psfrag{a4}{$a_4$}
\psfrag{b1}{$b_1$}
\psfrag{b2}{$b_2$}
\psfrag{b3}{$b_3$}
\psfrag{b4}{$b_4$}
\includegraphics[width=4in]{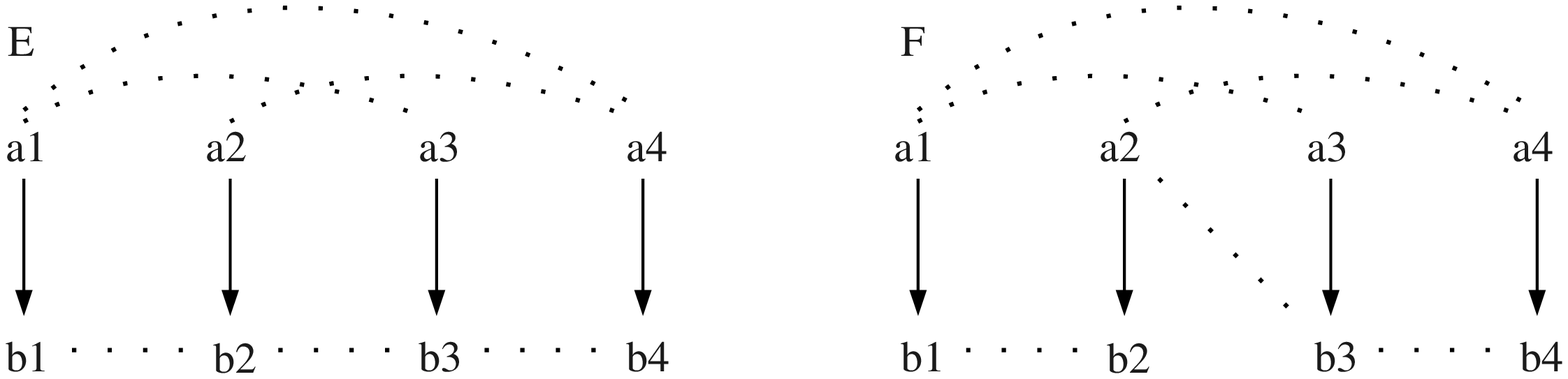}
\caption{Example~\ref{ex:hwh-not-h}.}\label{fig:hwh-not-h}
\end{figure}
To see that $\mc E$ and $\mc F$ are not H-equivalent,
consider $\emptyset \tran {a_2} \tran {a_3} \{a_2,a_3\}$ in $\mc F$.
This must be matched by moving to configuration $\{a_i,a_{i+1}\}$ in $\mc E$, 
where
$i\in\{1,2,3\}$.
But then both $b_i$ and $b_{i+1}$ are possible.
However $\{a_2,a_3\}$ in $\mc F$ can only do $b_2$.
Hence one of the $b_i$ and $b_{i+1}$ in $\mc E$ cannot be matched to $b_2$
in such way that the resulting isomorphism contains the already established 
pairs (either $(a_2,a_i),(a_3,a_{i+1})$ or $(a_2,a_{i+1}),(a_3,a_{i})$)
and is history-preserving.
\end{example}

\begin{example}\label{ex:absorption}
The Absorption Law \cite{BC87,Bed91,vGG01}
\[
(a \Par (b + c)) + (a \Par b) + ((a + c) \Par b) =
(a \Par (b + c)) + ((a + c) \Par b)
\]
holds for $\eqb{h}\,$, and thus for $\eqb{wh}\,$, but not for $\eqb{hwh}\,$.
\end{example}

\section{Event Identifier Logic}\label{sec:EIL}

We now introduce our logic, which we call Event Identifier Logic ($\EIL{}$).
We assume an infinite set of identifiers $\Id$, ranged over by $x,y,z,\ldots$.
The syntax of $\EIL{}$ is as follows:
\[
\phi \bnfeq \ttrue  \bsep \neg\phi \bsep \phi \aand \phi' \bsep \fd{x:a}\phi \bsep \dec{x:a}\phi \bsep \rd{x}\phi
\]
We include the usual operators of propositional logic: truth $\ttrue$, negation $\neg\phi$ and conjunction $\phi \aand \phi'$.
We then have {\em forward diamond} $\fd {x:a} \phi$, which says that it is possible to perform an event labelled with $a$ and reach a new configuration where $\phi$ holds.
In the formula $\fd {x:a} \phi$, the modality $\fd {x:a}$ binds all free occurrences of $x$ in $\phi$.
Next we have {\em declaration} $\dec {x:a} \phi$.
This says that there is some event with label $a$ in the current configuration which can be bound to $x$, in such a way that $\phi$ holds.
Here the declaration $\dec {x:a}$ binds all free occurrences of $x$ in $\phi$.
Finally we have {\em reverse diamond} $\rd x \phi$.
This says that it is possible to perform the reverse event bound to identifier $x$, and reach a configuration where $\phi$ holds.
Note that $\rd x$ does not bind $x$.
Clearly any occurrences of $x$ that get bound by $\dec{x:a}$ must be of the form $\rd x$.
We allow alpha-conversion of bound names.
We use $\phi,\psi,\ldots$ to range over formulas of $\EIL{}$.
\begin{example}
The formula $\fd{x:a} \fd{y:a} \rd x \ttrue$ says that there are events with label $a$,
say $e_1$ and $e_2$,  
that can be bound to $x$ and $y$ such that, after performing $e_1$ and then $e_2$, we can reverse $e_1$. Obviously, after performing $e_1$ followed by $e_2$, we can always reverse $e_2$. This formula could be interpreted as saying that an event bound to $x$ 
is \emph{concurrent} with an
event bound to $y$. Next, consider $\fd{x:a} \fd{y:a} \neg \rd x \ttrue$. The formula
expresses that an event bound to $x$ \emph{causes} an event bound to $y$ (because 
if we could reverse $x$ before $y$, we would reach a configuration containing $y$ and
not $x$, which contradicts $x$ being a cause of $y$).

\end{example}

\begin{definition}\label{def:fi}
We define $\Fi\phi$, the set of free identifiers of $\phi$, by induction on formulas:.
\[
\begin{array}{lll}
\Fi{\ttrue} = \emptyset
& \Fi{\phi_1 \aand \phi_2} = \Fi{\phi_1} \union \Fi{\phi_2}
& \Fi{\dec{x:a}\phi} = \Fi\phi \setminus \{x\}
\\
\Fi{\neg\phi} = \Fi\phi
& \Fi{\fd{x:a}\phi} = \Fi\phi \setminus \{x\}
& \Fi{\rd x \phi} = \Fi\phi \union \{x\}
\\
\end{array}
\]
We say that $\phi$ is {\em closed} if $\Fi\phi = \emptyset$; otherwise $\phi$ is {\em open}.
\end{definition}
In order to assign meaning to open formulas, as usual we employ environments which tell us what events the free identifiers are bound to.
\begin{definition}\label{def:envt}
An {\em environment} $\rho$ is a partial mapping from $\Id$ to events.
We say that {\em $\rho$ is a permissible environment for $\phi$ and $X$} if
$\Fi\phi \subseteq \dom{\rho}$ and 
$\rge{\rho \res \Fi\phi} \subseteq X$.

We let $\emptyset$ denote the empty environment.
We let $\rho[x \mapsto e]$ denote the environment $\rho'$ which agrees with $\rho$ except possibly on $x$,
where $\rho'(x) = e$ (and $\rho(x)$ may or may not be defined).
We abbreviate $\emptyset[x \mapsto e]$ by $[x \mapsto e]$.
We let $\rho \setminus x$ denote $\rho$ with the assignment to $x$ deleted
(if defined in $\rho$).
\end{definition}
Now we can formally define the semantics of $\EIL{}$:
\begin{definition}\label{def:sat}
Let $\mc C$ be a stable configuration structure.
We define a satisfaction relation $\mc C, X, \rho \models \phi$
where $X$ is a configuration of $\mc C$,
and $\rho$ is a permissible environment for $\phi$ and $X$,
by induction on formulas as follows
(we suppress the $\mc C$ where it is clear from the context):
\begin{itemize}
\item
$X,\rho \models \ttrue$ always
\item
$X,\rho \models \neg \phi$ iff $X,\rho \not\models \phi$
\item
$X,\rho \models \phi_1 \aand \phi_2$ iff
$X,\rho \models \phi_1$
and $X,\rho \models \phi_2$
\item
$X,\rho \models \fd{x:a}\phi$ iff $\exists X',e$ such that
$X \tranc e X'$ with $\lab(e) = a$ and
$X',\rho[x \mapsto e] \models \phi$
\item
$X,\rho \models \dec{x:a}\phi$ iff $\exists e \in X$ such that
$\lab(e) = a$ and
$X,\rho[x \mapsto e] \models \phi$
\item
$X,\rho \models \rd x \phi$ iff $\exists X',e$ such that
$X \rtranc e X'$ with $\rho(x) = e$
and $X',\rho \models \phi$
(and $\rho$ is a permissible environment for $\phi$ and $X'$)
\end{itemize}
For closed $\phi$ we further define $\mc C, X \models \phi$ iff $\mc C, X, \emptyset \models \phi$,
and $\mc C \models \phi$ iff $\mc C, \emptyset \models \phi$.
\end{definition}
In the case of $\rd x \phi$, note that even though according to the syntax $x$ is allowed to occur free in $\phi$, if $x$ does occur free in $\phi$ then $X,\rho \models \rd x \phi$ can never hold: if $\rho(x) = e$ and $X \rtranc e X'$ then $X',\rho \models \phi$ cannot hold, since $\rho$ is not a permissible environment for $\phi$ and $X'$, as $\rho$ assigns a free identifier of $\phi$ to an event outside $X'$.

\begin{example}
Consider the configuration structure from Example \ref{simple_pes}. The empty configuration
satisfies $\fd{x:a} \fd{y:a} \rd x \ttrue$: we have $\emptyset,\emptyset
\models \fd{x:a} \fd{y:a} \rd x \ttrue$ since $\{e_1,e_3\}, [x\mapsto e_1,y\mapsto e_3]
\models \rd x \ttrue$; the latter holds because $\{e_1,e_3\}\rtran{e_1}\{e_3\}$
and $\rho(x)=e_1$. Also, $\emptyset,\emptyset\models \fd{x:a} \fd{y:a} \neg \rd x \ttrue$. We have
$\emptyset,\emptyset \models \fd{x:a} \fd{y:a} \neg \rd x \ttrue$ since $\{e_1,e_2\}, 
[x\mapsto e_1,y\mapsto e_2] \models \neg \rd x \ttrue$. This is because $\{e_1,e_2\} 
\not\rtran{e_1} \{e_2\}$ as $\{e_2\}$ is not a configuration.

The closed formula $\dec{x:a}\ttrue$ says that there is some event labelled with $a$ in 
the current configuration: $X \models \dec{x:a}\ttrue$ iff $\exists e \in X.\ \lab(e) = a$.
Returning to Example \ref{simple_pes}, note that as well as $\{e_1,e_2\},
[x\mapsto e_1,y\mapsto e_2] \models \neg \rd x \ttrue$ this also holds: 
$\{e_1,e_2\}, [x\mapsto e_1,y\mapsto e_2] \models \dec{x:a} \rd x \ttrue$. By the definition 
of $\dec{x:a}$, the current environment is updated to $[x\mapsto e_2,y\mapsto e_2]$ and
 we obtain $\{e_1,e_2\}, [x\mapsto e_2,y\mapsto e_2] \models \rd x \ttrue$. Correspondingly,
$\{e_1,e_2\}, [x\mapsto e_1,y\mapsto e_2] \models \dec{x:a} \rd x \dec{y:a} \rd y \ttrue$.
However, $\{e_1,e_2\}, [x\mapsto e_1,y\mapsto e_2] \not\models \dec{x:a} \rd x \rd y \ttrue$
since $\{e_1\}, [x\mapsto e_2,y\mapsto e_2] \not\models \rd y \ttrue$.

\end{example}
We introduce further operators as derived operators of $\EIL{}$:
\begin{notation}[Derived operators]\label{not:abbrev}
Let $A = \{a_1,\ldots,a_n\}$ be a multiset of labels.
\begin{itemize}
\item
$\ffalse \Defeq \neg\ttrue$, \quad 
$\fb {x:a} \phi \Defeq \neg \fd{x:a}\neg\phi$, \quad
$\phi_1 \oor \phi_2 \Defeq \neg(\neg \phi_1 \aand \neg \phi_2)$
\item
Forward step
$\fd A \phi \Defeq
\fd{x_1:a_1}\cdots \fd{x_n:a_n}(\phi \, \aand \, \Aand_{i=1}^{n-1} \rd {x_i} \ttrue)$ 
where $x_1,\ldots,x_n$ are fresh and distinct (and in particular are not free in $\phi$).
We write $\fd {a_1,\ldots,a_n}\phi$ instead of $\fd {\{a_1,\ldots,a_n\}}\phi$.
In the case $n=1$ we have $\fd a \phi \Defeq \fd{x:a}\phi$ where $x$ is fresh.
\item
Reverse step
$\rd A \phi \Defeq
(x_1:a_1)\cdots (x_n:a_n) (\rd{x_1}\cdots\rd{x_n}\phi \, \aand \, \Aand_{i=2}^{n} \rd {x_i} \ttrue)$
where $x_1,\ldots,x_n$ are fresh and distinct (and in particular are not free in $\phi$).
We write $\rd {a_1,\ldots,a_n}\phi$ instead of $\rd {\{a_1,\ldots,a_n\}}\phi$.
In the case $n=1$ we have $\rd a \phi \Defeq \dec{x:a} \rd x \phi$ where $x$ is fresh.
\end{itemize}
\end{notation}

\begin{example}
Consider $\mc E$, $\mc F$ in Figure~\ref{fig:hwh-not-h} and $\phi\equiv
\fb{x:a} \fb{y:a} (\fd{z:b} \neg \rd{x}\ttrue \aand \fd{w:b}\neg \rd{y}\ttrue)$. 
We easily check
that $\mc E$ satisfies $\phi$ and $\mc F$ does not.
Next, consider $\psi\equiv \fd{x:a} (\fb{w:c}\ffalse \aand \fd{y:b}\rd{x}\fb{z:c}\ffalse)$.
Then the LHS structure of the Absorption Law in Example~\ref{ex:absorption} satisfies
$\psi$ and the RHS does not. Strictly speaking, event identifiers are not necessary
to distinguish the two pairs of configuration structures. A formula with simple
label modalities $ \fd{a} (\fb{c}\ffalse \aand \fd{b}\rd{a}\fb{c}\ffalse)$ is sufficient
for the the Absorption Law, and $\mc E$, $\mc F$ in Figure~\ref{fig:hwh-not-h} can be
distinguished by a logic with pomset modalities (both reverse and forward) defined over
runs \cite{Che92,PLS94}.
\end{example}

\begin{example}\label{ex:h-not-ri-h-pes}
Consider $\mc E$, $\mc F$ in Figure~\ref{fig:h-not-ri-h-pes}. There is a non-binary
conflict among the three initial $a$-events (indicated by a dashed ellipsis) defined by
requiring that at most two of these events can appear in any configuration. 
$\mc E$ and $\mc F$ are H equivalent:  we define a bisimulation by relating configurations
of identically labelled events (including where $a_4$ is matched with $a_4'$) and check
that it is an H. The structures are also HWH equivalent. This time we define a bisimulation
between order isomorphic configurations (of which there only five isomorphism
classes: $\emptyset$, $\{a\}$,
$\{a,a\}$, $\{a<a\}$ and $\{a<a,a\}$, where events separated by commas are concurrent)
and check that it is an HWH. However, $\mc E$ and $\mc F$ are not HH equivalent and event
identifiers are indeed necessary to distinguish them. The formula
$\fd{x:a}\fd{y:a}(\neg \rd{x}\ttrue \aand \fd{z:a}\rd{y}\fd{w:a}\neg \rd{z}\ttrue \aand 
\fd{z':a}\rd{y}\neg \fd{w':a}\neg \rd{z'}\ttrue)$ is only satisfied by $\mc E$.  
It requires that $x$ causes $y$ and that $z$ and $z'$ are bound to different
events because $\fd{z:a}$ and $\fd{z':a}$ are followed by mutually
contradictory behaviours. 
This is possible in $\mc E$ ($a_1,a_4$ can be followed by either $a_3$ or $a_2$) 
but not in $\mc F$: none of the pairs of causally dependent events offers 
two different $a$-events.
\begin{figure}
\centering
\psfrag{E}{$\mc E$}
\psfrag{F}{$\mc F$}
\psfrag{a1}{$a_1$}
\psfrag{a2}{$a_2$}
\psfrag{a3}{$a_3$}
\psfrag{a4}{$a_4$}
\psfrag{a5}{$a_5$}
\psfrag{a6}{$a_6$}
\psfrag{a7}{$a_4'$}
\includegraphics[width=4in]{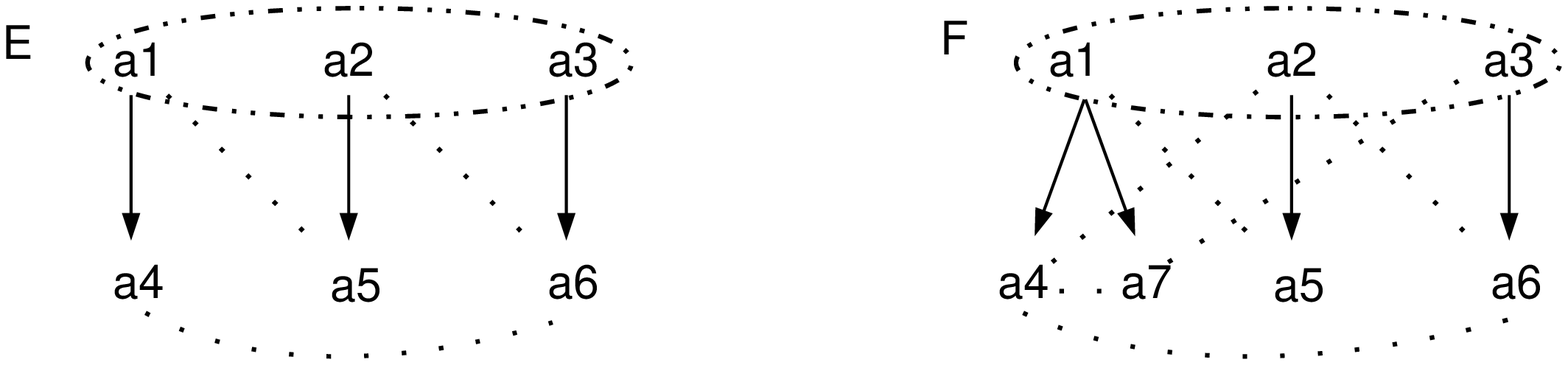}
\caption{Example~\ref{ex:h-not-ri-h-pes}.}\label{fig:h-not-ri-h-pes}
\end{figure}

\end{example}

\Comment{
\begin{lemma}\label{lem:env agree}
Let $X$ be a configuration of a stable configuration structure $\mc C$,
and let $\phi \in \EIL{}$.
Suppose that $\rho,\rho'$ are permissible environments for $\phi$ and $X$, which agree on $\Fi\phi$.
Then $X,\rho \models \phi$ iff $X,\rho' \models \phi$.
\end{lemma}
\begin{proof}
See Appendix~\ref{sec:proofs EIL}.
\end{proof}
Probably use Lemma~\ref{lem:env fresh} instead of the following:
\begin{lemma}\label{lem:new env fresh}
Let $X$ be a configuration of a stable configuration structure $\mc C$.
Let $\{z_e:e \in X\}$ be a set of distinct fresh identifiers,
and let $\rho_X$ be defined by $\rho_X(z_e) = e$.
Then for any $\phi \in \EIL{}$ and any $\rho$ a permissible environment for $\phi$ and $X$,
we have $X,\rho \models \phi$ iff $X,\rho_X \models \hat\phi$,
where $\hat\phi$ is got from $\phi$ by replacing each occurrence of a free identifier $x$ by $z_{\rho(x)}$.
\end{lemma}
Note that in Lemma~\ref{lem:new env fresh} the environment $\rho_X$ is injective,
but $\rho$ need not be injective, so that distinct free identifiers in $\phi$ may be mapped to the same identifier in $\hat\phi$.
}

\section{Using EIL to characterise equivalences}\label{sec:EIL eqnces}

We wish to show that $\EIL{}$ and its various sublogics characterise the equivalences defined in Section~\ref{subsec:equivalences}.
Each sublogic of $\EIL{}$ induces an equivalence on configuration structures in a standard 
fashion:
\begin{definition}
Let $L$ be any sublogic of $\EIL{}$.
Then $L$ induces an equivalence on stable configuration structures as follows:
$\mc C \eql{L} \mc D$ iff
for all closed $\phi \in L$ we have $\mc C \models \phi$ iff $\mc D \models \phi$.
\end{definition}
First we introduce a simple sublogic that allows us to characterise order isomorphism.
\subsection{Reverse-only logic and order isomorphism}\label{subsec:EILro}

We define sublogics of $\EIL{}$, consisting of formulas where only reverse transitions are allowed.
\begin{definition}\label{def:EILro}
{\em Reverse-only} logic $\EIL{ro}$:
\[
\phi \bnfeq \ttrue \bsep \neg \phi \bsep \phi \aand \phi' \bsep \dec{x:a}\phi \bsep \rd{x} \phi 
\]
We further define {\em declaration-free} reverse-only logic $\EIL{dfro}$:
\[
\phi \bnfeq \ttrue \bsep \neg \phi \bsep \phi \aand \phi' \bsep \rd{x} \phi 
\]
\end{definition}
These logics are preserved between isomorphic configurations,
and characterise configurations up to isomorphism.
\begin{lemma}\label{lem:ro isom}
Let $\mc C,\mc D$ be stable configuration structures,
and let $X,Y$ be configurations of $\mc C, \mc D$ respectively.
Suppose that $f: X \isom Y$.
Then for any $\phi \in \EIL{ro}$, and any $\rho$ (permissible environment for $\phi$ and $X$),
we have $X, \rho \models \phi$ iff $Y, f \circ \rho_\phi \models \phi$.
\end{lemma}
Recall that $\rho_\phi$ is an abbreviation for $\rho\res\Fi\phi$.
Function composition is in applicative rather than diagrammatic order.

Given any configuration $X$ we can create a closed formula $\theta_X \in \EIL{ro}$ which gives the order structure of $X$.
We make this precise in the following lemma:
\begin{lemma}\label{lem:ro X}
Let $X$ be a configuration of a stable configuration structure $\mc C$.
There is a {\em closed} formula $\theta_X \in \EIL{ro}$,
such that
if $Y$ is any configuration of a stable configuration structure $\mc D$
and $\card Y = \card X$, then $Y \isom X$ iff 
$Y \models \theta_X$.
\end{lemma}
The next lemma follows fairly immediately from the proof of Lemma~\ref{lem:ro X} and from Lemma~\ref{lem:ro isom}:
\begin{lemma}\label{lem:dfro X}
Let $X$ be a configuration of a stable configuration structure $\mc C$.
Let $\{z_e :e \in X\}$ be distinct identifiers.
Let the environment $\rho_X$ be defined by $\rho_X(z_e) = e$ ($e \in X$).
There is a formula $\theta'_X \in \EIL{dfro}$ with $\Fi{\theta'} = \{z_e :e \in X\}$,
such that $X,\rho_X \models \theta'_X$ and
if $Y$ is any configuration of a stable configuration structure $\mc D$
and $\card Y = \card X$, then $Y \isom X$ iff 
$\exists \rho.\;Y, \rho \models \theta'_X$.
\end{lemma}

\subsection{Logics for history-preserving bisimulations}\label{subsec:hist}

We start by showing that $\EIL{}$ characterises HH-bisimulation.
We then present sublogics of $\EIL{}$ which correspond to H-bisimulation, WH-bisimulation 
and HWH-bisimulation. 

Our first result is related to the result of~\cite{NC94a} that a logic with reverse event index modality (discussed above in Section~\ref{sec:related}) characterises HH.
\begin{theorem}\label{thm:HH=EIL}
Let $\mc C,\mc D$ be stable configuration structures.
Then, $\mc C \eqb{hh} \mc D$ if and only if $\mc C \eql{\EIL{}} \mc D$.
\end{theorem}
\Comment{
\begin{proof}[Proof sketch]
($\Rightarrow$)
Let $\mathcal{R}$ be an HH bisimulation between $\mc C$ and $\mc D$.
We show by induction on $\phi$ that for all $X,Y,f$,
if $\mathcal{R}(X,Y,f)$ then for all $\phi \in \EIL{}$
and all $\rho$ (permissible environment for $\phi$ and $X$) we have 
$X, \rho \models \phi$ iff $Y, f \circ \rho_\phi \models \phi$.
Recall that $\rho_\phi$ is an abbreviation for $\rho\res\Fi\phi$.
By considering initial (empty) configurations,
our induction hypothesis implies that $\mc C \eql{\EIL{}} \mc D$.

($\Leftarrow$)
Suppose that $\mc C \eql{\EIL{}} \mc D$.
Define $\mathcal{R}(X,Y,f)$ iff
\begin{itemize}
\item
$f$ is an order isomorphism between $X$ and $Y$
\item
for any $\phi \in \EIL{}$ and $\rho$ (permissible environment for $\phi$ and $X$) with $\rge\rho \subseteq X$ we have
$X, \rho \models \phi$ iff $Y, f \circ \rho \models \phi$.
(Note that by considering negated formulas,
$X, \rho \models \phi$ iff $Y, f \circ \rho \models \phi$
is equivalent to
$X, \rho \models \phi$ implies $Y, f \circ \rho \models \phi$.)

\end{itemize}
We show that $\mathcal R$ is an HH bisimulation.
For the case where $X\tranc{a} X'$,
we need $Y\tranc{a} Y'$ with $\mathcal{R}(X',Y',f)$.
We make use of the open formula $\theta'_{X'}$ of Lemma~\ref{lem:dfro X}
to ensure that $Y' \isom X'$.
\end{proof}
}
\begin{remark}
In fact Theorem~\ref{thm:HH=EIL} would hold with the logic restricted by not using declarations $\dec{x:a}\phi$.
However we include declarations in EIL because they are useful in defining sublogics for WH, among other things.
\end{remark}

We define a sublogic of $\EIL{}$ which characterises history-preserving bisimulation:
\begin{definition}\label{def:EILh}
$\EIL{h}$ is given as follows, where $\phi_r$ is a formula of  $\EIL{ro}$:
\[
\phi \bnfeq \ttrue \bsep \neg \phi \bsep \phi \aand \phi' \bsep \fd{x:a}\phi \bsep \dec{x:a}\phi \bsep \phi_r 
\]
\end{definition}
$\EIL{h}$ is just $\EIL{}$ with $\rd {x:a} \phi$ replaced by $\phi_r \in \EIL{ro}$.
Thus one is not allowed to go forward after going in reverse.
This concept of disallowing forward moves embedded inside reverse moves appears in~\cite{GKP92}.

\begin{theorem}\label{thm:H=EILh}
Let $\mc C,\mc D$ be stable configuration structures.
Then, $\mc C \eqb{h} \mc D$ if and only if $\mc C \eql{\EIL{h}} \mc D$.
\end{theorem}
\Comment{
\begin{proof}[Proof sketch]
($\Rightarrow$)
We follow the proof of Theorem~\ref{thm:HH=EIL},
except that when dealing with $\phi_r \in \EIL{ro}$ instead of using the main induction hypothesis, we use Lemma~\ref{lem:ro isom}.
($\Leftarrow$) Follow the proof of Theorem~\ref{thm:HH=EIL}.
\end{proof}
}
\begin{remark}
Just as for Theorem~\ref{thm:HH=EIL}, Theorem~\ref{thm:H=EILh} would still hold if we 
disallow declarations $\dec{x:a}\phi$.  This gives the following more minimal logic,
where $\phi_r \in \EIL{dfro}$.
\[
\phi \bnfeq \ttrue \bsep \neg \phi \bsep \phi \aand \phi' \bsep \fd{x:a}\phi \bsep \phi_r 
\]
\end{remark}


We define a sublogic $\EIL{wh}$ of $\EIL{h}$ which characterises weak history-preserving bisimulation.
We get from $\EIL{h}$ to $\EIL{wh}$ by simply requiring that all formulas of $\EIL{wh}$ are {\em closed}.
\begin{definition}\label{def:EILwh}
$\EIL{wh}$ is given as follows, where $\phi_{rc}$ is a {\em closed} formula of  $\EIL{ro}$ (Definition~\ref{def:EILro}):
\[
\phi \bnfeq \ttrue \bsep \neg \phi \bsep \phi \aand \phi'\bsep \fd a \phi \bsep \phi_{rc}
\]
\end{definition}
In the above definition we write $\fd a \phi$ rather than $\fd{x:a}\phi$ since $\phi$ is closed and in particular $x$ does not occur free in $\phi$ (Notation~\ref{not:abbrev}).
Also we omit declarations $\dec{x:a}\phi$ since they have no effect when $\phi$ is closed.
Of course declarations can occur in $\phi_{rc}$.

\begin{theorem}\label{thm:WH=EILwh}
Let $\mc C,\mc D$ be stable configuration structures.
Then, $\mc C \eqb{wh} \mc D$ iff $\mc C \eql{\EIL{wh}} \mc D$.
\end{theorem}
\Comment{
\begin{proof}[Proof sketch]
Apart from the use of Lemmas~\ref{lem:ro isom} and~\ref{lem:ro X} to handle formulas of 
$\EIL{ro}$, the proof is much as for standard Hennessy-Milner logic~\cite{HM85}.
\end{proof}
}

We believe that $\EIL{wh}$ is the first logic proposed for weak history-preserving bisimulation with autoconcurrency allowed.
Goltz et al.\ \cite{GKP92} described a logic for weak history-preserving bisimulation with no autoconcurrency allowed, but in this case, weak history-preserving bisimulation is as strong as history-preserving bisimulation~\cite{vGG01}.

Just as we weakened $\EIL{h}$ to get $\EIL{wh}$ we can weaken $\EIL{}$
by requiring that
forward transitions $\fd{x:a}\phi$ are only allowed if $\phi$ is closed.
Again instead of $\fd{x:a}\phi$ we write $\fd a \phi$.
This gives us $\EIL{hwh}$:
\begin{definition}
$\EIL{hwh}$ is given below, where $\phi_c$ ranges over closed formulas of $\EIL{hwh}$. 
\[
\phi \bnfeq \ttrue \bsep \neg \phi \bsep \phi \aand \phi' \bsep \fd{a} \phi_c \bsep \dec{x:a}\phi \bsep \rd{x}\phi
\]
\end{definition}
Plainly $\EIL{wh}$ is a sublogic of $\EIL{hwh}$ as well as of $\EIL{h}$.
\begin{theorem}\label{thm:HWH=EILhwh}
Let $\mc C,\mc D$ be stable configuration structures.
Then, $\mc C \eqb{hwh} \mc D$ iff $\mc C \eql{\EIL{hwh}} \mc D$.
\end{theorem}

With no (equidepth) autoconcurrency, we know that $\eqb{hwh}$ is as strong as $\eqb{hh}$~\cite{Bed91,PU11}.
So $\EIL{hwh}$ is as strong as $\EIL{}$ in this case.

\section{Characteristic formulas}\label{sec:char}

In this section we investigate characteristic formulas for three of the equivalences we have considered, namely HH, H and WH.
The idea is that we reduce checking whether $\mc C$ and $\mc D$ satisfy the same formulas in a logic such as $\EIL{}$ to the question of whether $\mc D$ satisfies a particular formula $\chi_{\mc C}$, the {\em characteristic formula} of $\mc C$,
which completely expresses the behaviour of $\mc C$, at least as far as the particular logic is concerned.
As pointed out in~\cite{AIS09}, this means that checking whether two structures are equivalent is changed from the problem of potentially having to check infinitely many formulas into a single model-checking problem $\mc D \models \chi_{\mc C}$.

Characteristic formulas for models of concurrent systems were first investigated in~\cite{GS86}, and subsequently in~\cite{SI94} and other papers---see~\cite{AIS09} for further references.
As far as we are aware, characteristic formulas have not previously been investigated for any true concurrency logic, although we should mention that
in~\cite{AIS09} characteristic formulas are studied for a logic with both forward and reverse modalities, related to the back and forth simulation of~\cite{DNMV90}.

We shall confine ourselves to {\em finite} stable configuration structures in this section.
Even with this assumption, it is not obvious that an equivalence such as HH, which employs both forward and reverse transitions, can be captured by a single finite-depth formula.  To show that forward and reverse transitions need not alternate for ever, we first relate HH to a simple game.

\begin{definition}
Let $\mc C,\mc D$ be finite stable configuration structures.
The game $G(\mc C,\mc D)$ has two players: $A$ (attacker) and $D$ (defender).
The set of game states is
\(
S(\mc C,\mc D) \Defeq \{(X,Y,f):
X \in C_{\mc C},Y \in C_{\mc D},f:X \isom Y\}
\).
The start state is $(\emptyset,\emptyset,\emptyset)$.
At each state of the game $A$ chooses a forward (resp.\ reverse) move $e$ of
either $\mc C$ or $\mc D$.
Then $D$ must reply with a corresponding forward (resp.\ reverse) move $e'$ by the other structure.
Going forwards we extend $f$ to $f'$ and going in reverse we restrict $f$ to $f'$,
as in the definition of HH.
The two moves produce a new game state $(X',Y',f')$.
Then $D$ wins if we get to a previously visited state.
Conversely, $A$ wins if $D$ cannot find a move.
(Also $D$ wins if $A$ cannot find a move, but that can only happen if both
$\mc C$ and $\mc D$ have only the empty configuration.)
\end{definition}
It is reasonable that $D$ wins if a state is repeated,
since if $A$ then chooses a different and better move at the repeated state,
$A$ could have chosen that on the previous occasion.
\begin{definition}
Given finite stable configuration structures $\mc C,\mc D$,
let $s(\mc C,\mc D) \Defeq \card{S(\mc C,\mc D)}$,
let $c(\mc C) = \max\{\card X : X \in C_{\mc C}\}$,
and let $c(\mc C,\mc D) = \min\{c(\mc C),c(\mc D)\}$.
\end{definition}
\Comment{
The set $S(\mc C,\mc D)$ can be viewed as an undirected graph 
where $(X,Y,f)$ is adjacent to $(X',Y',f')$
if $X' = X \union \{e\}$, $Y' = Y \union \{e'\}$ and $f' = f \union \{(e,e')\}$.
A play in $G(\mc C,\mc D)$ corresponds to a path in $S(\mc C,\mc D)$
with start node $(\emptyset,\emptyset,\emptyset)$,
and with no repeated node apart possibly from the last node.
}
Clearly any play of the game $G(\mc C,\mc D)$ finishes after no more than $s(\mc C,\mc D)$ moves.
We can place an upper bound on $s(\mc C,\mc D)$ as follows:
\begin{proposition}
Let $\mc C,\mc D$ be finite stable configuration structures.
Then $s(\mc C,\mc D) \leq \card{C_{\mc C}}.\card{C_{\mc D}}.c(\mc C,\mc D)!$.
\end{proposition}
Note that if there is no autoconcurrency, any isomorphism $f:X \isom Y$ is unique,
and so we can improve the upper bound on the number of states to
$s(\mc C,\mc D) \leq \card{C_{\mc C}}.\card{C_{\mc D}}$.
\begin{proposition}\label{prop:hh game}
Let $\mc C,\mc D$ be finite stable configuration structures.
Then $\mc C \eqb{hh} \mc D$ iff defender $D$ has a winning strategy for the game $G(\mc C,\mc D)$.
\end{proposition}

\begin{remark}
Certainly game characterisations of HH equivalence have been used many times before;
see e.g.\ \cite{Fro99,Fro05,FL05,JNS03,Gut09}.
However defender is usually said to win if the play continues for ever,
whereas we say that defender wins if a state is repeated.
This is because we are working with finite configuration structures, rather than, say, Petri nets.
\end{remark}
\begin{definition}
Let $\phi \in \EIL{}$.
The {\em modal depth} $\depth{\phi}$ of $\phi$ is defined as follows:
\[
\begin{array}{lll}
\depth{\ttrue} \Defeq 0
&
\depth{\phi \aand\phi'} \Defeq \max(\depth{\phi},\depth{\phi'})
&
\depth{\dec{x:a}\phi} \Defeq \depth{\phi}
\\
\depth{\neg\phi} \Defeq \depth{\phi}
&
\depth{\fd{x:a}\phi} \Defeq 1+\depth{\phi}
&
\depth{\rd{x:a}\phi} \Defeq 1+\depth{\phi}
\\
\end{array}
\]
\end{definition}
We can use the game characterisation of HH to bound the modal depth of $\EIL{}$ formulas needed to check whether finite structures are HH equivalent:
\begin{theorem}\label{thm:hh depth}
Let $\mc C,\mc D$ be finite stable configuration structures.
Then $\mc C \eqb{hh} \mc D$
iff $\mc C$ and $\mc D$ satisfy the same $\EIL{}$ formulas of modal depth no more than $s(\mc C,\mc D)+c(\mc C,\mc D)$.
\end{theorem}
We now define a family of characteristic formulas for HH equivalence,
parametrised on modal depth.
\begin{definition}
Suppose that $\Act$ is finite.
Let $\mc C$ be a finite stable configuration structure.
We define formulas $\Chard{hh}{X}{n}$ ($X$ a configuration of $\mc C$) by induction on $n$:
\[
\begin{array}{rcl}
\Chard{hh}{X}{0} &\Defeq & \theta'_X \\
\Chard{hh}{X}{n+1} &\Defeq
& \theta'_X
\displaystyle
\aand (\Aand_{X \tranc e X'} \fd {z_e:\lab(e)} \Chard{hh}{X'}{n})
\aand {(\Aand_{a \in \Act} \fb{x:a} \Oor_{X \tranc e X',\lab(e)=a} \Chard{hh}{X'}{n}[x/z_e])}
\displaystyle
\aand ( {\Aand_{X \rtranc e X'} \rd {z_e} \Chard{hh}{X'}{n}} ) \\
\end{array}
\]
Here $\theta'_X \in \EIL{dfro}$ is as in Lemma~\ref{lem:dfro X}
and $\Fi{\Chard{hh}{X}{n}} = \{z_e : e \in X\}$.
We further let $\Chard{hh}{\mc C}{n} \Defeq \Chard{hh}{\emptyset}{n}$.
\end{definition}
Note that $\Chard{hh}{X}{n} \in \EIL{}$ and $\depth{\Chard{hh}{X}{n}} \leq n+c(\mc C)$.
\begin{theorem}\label{thm:char hh}
Suppose that $\Act$ is finite.
Let $\mc C,\mc D$ be finite stable configuration structures.
Let $s \Defeq s(\mc C,\mc D)$.
Then $\mc C \eqb{hh} \mc D$
iff $\mc D \models \Chard{hh}{\mc C}{s}\;$.
\end{theorem}
Thus we do not have a single characteristic formula for $\mc C$,
but we can deal uniformly with all $\mc D$ up to a certain size.
This is almost as good as having a single characteristic formula for $\mc C$,
since we can generate a formula of the appropriate size once we have settled on $\mc D$,
so that we have still reduced equivalence checking to checking a single formula.
Single characteristic formulas are certainly possible for some $\mc C$s;
there remains an open question of whether for all finite $\mc C$
there is a single formula $\Char{hh}{\mc C}$
which works for all $\mc D$.

\Comment{
We still need to find a counter-example to show that $\mc C$ need not have a single
characteristic formula which works for all $\mc D$.
We would need $\mc C$ and a whole family of $\mc D_n$
such that $\mc D_n \models \Chard{hh}{\mc C}{m}$ for small $m$, but not if $m$ is sufficiently large.
This is a major gap in our knowledge, since it could be that there is a single characteristic formula for every structure.
}

Matters are simpler for H and WH equivalences,
since only forward transitions are employed.
\begin{definition}
Suppose that $\Act$ is finite.
Let $\mc C$ be a finite stable configuration structure.
We define formulas $\Char{h}{X}$ ($X$ a configuration of $\mc C$) as follows:
\[
\Char{h}{X} \Defeq \theta'_X \aand
( \Aand_{X \tranc e X'} \fd {z_e:\lab(e)} \Char{h}{X'} )
\aand ( \Aand_{a \in \Act}\fb {x:a} \Oor_{X \tranc e X',\lab(e)=a} \Char{h}{X'}[x/z_e] )
\]
Here $\theta'_X \in \EIL{dfro}$ is as in Lemma~\ref{lem:dfro X}.
We further let $\Char{h}{\mc C} \Defeq \Char{h}{\emptyset}$.
\end{definition}
Note that $\Char{h}{\mc C} \in \EIL{h}$; it is well-defined,
since maximal configurations form the base cases of the recursion.
Also $\depth{\Char{h}{X}} \leq 2.c(\mc C)$.

\begin{proposition}\label{prop:char h}
Suppose that $\Act$ is finite.
Let $\mc C,\mc D$ be finite stable configuration structures.
Then $\mc D \eqb{h} \mc C$ iff $\mc D \models \Char{h}{\mc C}\;$.
\end{proposition}
WH is even easier as formulas are closed:
\begin{definition}
Suppose that $\Act$ is finite.
Let $\mc C$ be a finite stable configuration structure.
We define formulas $\Char{wh}{X}$ ($X$ a configuration of $\mc C$) as follows:
\[
\Char{wh}{X} \Defeq \theta_X
\ \aand \  (\Aand_{X \tranc a X'} \fd a \Char{wh}{X'})
\ \aand \  (\Aand_{a \in \Act}\fb a \Oor_{X \tranc a X'} \Char{wh}{X'})
\]
Here $\theta_X \in \EIL{ro}$ is as in Lemma~\ref{lem:ro X}.
We further let $\Char{wh}{\mc C} \Defeq \Char{wh}{\emptyset}$.
\end{definition}
Note that $\Char{wh}{\mc C} \in \EIL{wh}$ and $\depth{\Char{wh}{X}} \leq 2.c(\mc C)$.
\begin{proposition}\label{prop:char wh}
Suppose that $\Act$ is finite.
Let $\mc C,\mc D$ be finite stable configuration structures.
Then $\mc D \eqb{wh} \mc C$ iff $\mc D \models \Char{wh}{\mc C}\;$.
\end{proposition}

\section{Conclusions and future work}

We have introduced a logic which uses event identifiers to track events in both forwards and reverse directions.  As we have seen, this enables it to express causality and concurrency between events.
The logic is strong enough to characterise hereditary history-preserving (HH) bisimulation equivalence.  We are also able to characterise weaker equivalences using sublogics.  In particular we can characterise weak history-preserving bisimulation,
which has not been done previously as far as we are aware.  We also investigated characteristic formulas for our logic with respect to HH and other equivalences.
Again we are not aware of previous work on characteristic formulas for logics for true concurrency.

Baldan and Crafa~\cite{BC10} gave logics for pomset bisimulation and step bisimulation;
we have also been able to characterise these equivalences in our setting,
but we had to omit this material for reasons of space.

In future work we would like to (1) investigate general laws which hold for the logic,
(2) look at sublogics characterising other true concurrency equivalences,
including equivalences involving reverse transitions from~\cite{Bed91,PU11},
and (3) answer the open question raised in Section~\ref{sec:char} about whether there
is a single characteristic formula for a finite structure with respect to HH equivalence.
\Comment{
\begin{itemize}
	\item 
Look at natural fragments of $\EIL{}$.  Examples might include the negation-free fragment, or where Boolean operations are only on closed formulas.
	\item 
Investigate equational laws.
We can argue that it will be easier to find laws in $\EIL{}$ compared to Baldan \& Crafa's logic.  Our syntax seems simpler/more standard.
\item
Look at logics characterising reverse equivalences from~\cite{Bed91,PU11}, such as RI-IB (easy), RS-SB (easy), RP-PB, H-WHPB, etc.
\item
Answer the question of whether there is a single characteristic formula $\Char{hh}{\mc C}$.
\end{itemize}
}

\paragraph{Acknowledgements.}  We are grateful to Ian Hodkinson and the anonymous referees for helpful comments and suggestions.

\bibliographystyle{eptcs}

\end{document}